\documentclass[journal]{IEEEtran}  

\IEEEoverridecommandlockouts           
\usepackage{graphics} 
\usepackage{epsfig} 
\usepackage{amsmath} 
\usepackage{amssymb}  
\usepackage{amsthm}  
\usepackage{enumerate}
\newtheorem{theorem}{Theorem}
\newtheorem{remark}{Remark}
\usepackage{afterpage}
\usepackage{euscript}
\usepackage{multirow}
\usepackage{xparse}
\usepackage{cite}
\newcommand{\tp}{\mathrm{T}}

\newenvironment{changemargin}[2]{%
\begin{list}{}{%
\setlength{\topsep}{0pt}%
\setlength{\leftmargin}{#1}%
\setlength{\rightmargin}{#2}%
\setlength{\listparindent}{\parindent}%
\setlength{\itemindent}{\parindent}%
\setlength{\parsep}{\parskip}%
}%
\item[]}{\end{list}}

\usepackage{mathtools}
\title{\LARGE \bf
Constrained Inverse Optimal Control with Application to a \\ 
Human Manipulation Task
}

\author{Marcel Menner, Peter Worsnop, and Melanie N. Zeilinger
        \thanks{M. Menner, P. Worsnop, and M.N. Zeilinger 
        are with the Institute for Dynamic Systems and Control,
       ETH Zurich, 8092 Zurich, Switzerland   \hspace{3cm}
        {\tt\small \{mmenner,pworsnop,mzeilinger\}@ethz.ch}}%
\thanks{This work was supported by the Swiss National Science Foundation
under grant no. PP00P2${\_}$157601 / 1.}
}

\begin{document}

\titlepage
\vspace*{3cm}
\begin{center}
\LARGE \bf
Constrained Inverse Optimal Control with Application to a \\ 
Human Manipulation Task
\end{center}
\vspace{0.0cm}
\begin{center}
Marcel Menner, Peter Worsnop, and Melanie N. Zeilinger
\end{center}

\vspace{3cm}
\noindent This work has been accepted for publication in the IEEE Transactions on Control Systems Technology.

\noindent Digital Object Identifier 10.1109/TCST.2019.2955663.

\vspace{1cm}
\noindent Please cite this work as
\\

\noindent  M. Menner, P. Worsnop, and M.N. Zeilinger, "Constrained Inverse Optimal Control with Application to a Human Manipulation Task," \textit{IEEE Transactions on Control Systems Technology}, 2019, doi:10.1109/TCST.2019.2955663.
\vspace{10cm}

\begin{changemargin}{0cm}{0cm}
\noindent 
\textcopyright\hspace{-0.1cm} 2019 IEEE.  
 Personal use of this material is permitted.  Permission from IEEE must be obtained for all other uses, in any current or future media, including reprinting/republishing this material for advertising or promotional purposes, creating new collective works, for resale or redistribution to servers or lists, or reuse of any copyrighted component of this work in other works.
\end{changemargin}

\twocolumn
\newpage

\maketitle
\thispagestyle{empty}
\pagestyle{empty}

\begin{abstract}
This paper presents an inverse optimal control methodology and its application to training a predictive model of human motor control from a manipulation task. 
It introduces a convex formulation for  learning both objective function and constraints of an infinite-horizon constrained optimal control problem with nonlinear system dynamics.
The inverse approach utilizes Bellman's principle of optimality to formulate the infinite-horizon optimal control problem as a shortest path problem and Lagrange multipliers to identify constraints. 
We highlight the key benefit of using the shortest path formulation, i.e., the possibility of training the predictive model with short and selected trajectory segments.
The method is applied to training a predictive model of movements of a human subject from a manipulation task.
The study indicates that individual human movements can be predicted with low error using an infinite-horizon optimal control problem with constraints on shoulder movement. 
\end{abstract}
\section{Introduction}
\IEEEPARstart{A}{s robotic}
systems are applied to increasingly unstructured and unpredictable environments, the ability to identify and adapt to their environment is becoming of critical importance. 
The collaboration with humans represents a particular challenge, as the interaction varies between individuals.
The manipulation of an articulated object by a human in collaboration with a robot is one example, where the robot performance can be improved by learning a model to describe and predict the human motor control behavior \cite{Lee2012}.

The literature on human control behavior widely agrees on the fact that human motor performance is achieved through the reactive and predictive component (see the review in \cite{Wolpert2011}).
The reactive component is triggered by sensory inputs and updates an ongoing motor command; it can, therefore, be interpreted as the feedback control action. 
The predictive component capitalizes on the ability to anticipate motor events based on memory in order to accomplish a given task under foreseeable conditions, which can be interpreted as feedforward action \cite{Kawato1999}.
The existence of these two components has been highlighted in studies of various motor control tasks, including grasping and manipulation \cite{Johansson1988,Johansson1992,Fu2010}.

In this work, we present a shortest path inverse optimal control method, which is applied to train a predictive model of human motor control. 
The inverse optimal control method is thereby used to learn the parameters of an optimal control problem from demonstrated state and input trajectories.
In particular, it learns both the objective function and constraints of an underlying infinite-horizon optimal control problem from observed trajectory segments of finite length using optimality conditions of a corresponding shortest path problem and a candidate constraint set.
The optimality conditions are derived based on Bellman's principle of optimality \cite{Bellman1957} and the \underline Karush-\underline{K}uhn-\underline{T}ucker (KKT) optimality conditions \cite{Kuhn1951}. 
The proposed method is convex for objective functions that are linear in their parameters and for general nonlinear systems, where relevant constraints are identified from the candidate constraint set using Lagrange multipliers.
The method is utilized to train a predictive model of movements of three human subjects from a human manipulation task. 

We set up a human manipulation experiment, where three human subjects manipulated one end of a passive kinematic object whose position was changed consecutively by a robot.
In this context, the goal of the inverse learning method is to train a predictive model of human movements. 
The underlying hypothesis is that the demonstrations of the human manipulation task are optimal with respect to an infinite-horizon constrained optimal control problem.
The experimental study highlights the potential of the proposed learning approach by providing good predictive performance for individual human movements.
In particular, the proposed shortest path formulation is shown to be beneficial for suboptimal execution, i.e., disregard the reactive human motor control component in the application considered in this paper.

Related inverse optimal control approaches are presented in \cite{Kalman1964, Priess2015,Menner2018,Mombaur2010,Puydupin2012,Englert2017,Majumdar2017}.
The approaches in \cite{Kalman1964,Priess2015,Menner2018} can be interpreted as an inverse method of an infinite-horizon optimal control problem, but they are restricted to unconstrained, linear systems and quadratic objective functions.
In \cite{Mombaur2010}, a bilevel approach to solve an inverse unconstrained optimal control problem is presented. 
The techniques closest to our method are \cite{Puydupin2012,Englert2017,Majumdar2017}, where the KKT conditions are similarly used for learning the stage cost but the constraints are assumed to be known. 
The two main distinctions of our approach with respect to \cite{Puydupin2012,Englert2017,Majumdar2017} are the consideration of an optimal control problem with an infinite horizon and the simultaneous identification of constraints from a candidate constraint set that is constructed from data with a convex optimization problem.
By using a shortest path formulation, the required trajectory segment for learning the parameters of the underlying optimal control problem can be shorter, e.g., compared to \cite{Englert2017}, and the learned parameters are invariant with respect to the chosen trajectory segment.
As for the application, the incorporation of constraints results in better predictions of human movement, whereas the consideration of a shortest path formulation allows for isolating trajectory segments where the predictive component is dominant, i.e., where the hypothesis of optimal demonstrations with respect to an optimal controller is valid.

\section{Shortest Path Inverse Optimal Control}
\label{sec:opt}

This section presents an \underline inverse \underline optimal \underline control (IOC) approach based on a shortest path formulation to learn an objective function and constraints from observations.
The observations are represented as trajectories of state measurements  $x(k)\in \mathbb{R}^n$ and inputs $u(k)\in \mathbb{R}^m$ at time-step $k$, where
\begin{equation}
\label{eq:sys}
x(k+1)=f(x(k),u(k))
\end{equation}
with the potentially nonlinear function $f(\cdot)$ modeling the evolution of the state. 
For the derivation of the inverse method in this section, we assume that $f(\cdot)$ is given. 
Section~\ref{sec:robot} discusses how to identify $f(\cdot)$ for the considered application. 

Observed trajectories are assumed to be optimal with respect to an infinite-horizon constrained optimal control problem, i.e., $x(k+i)=x_i^\star$ and $u(k+i)=u_i^\star$ $\forall$ $i\geq 0$ with 
\begin{subequations}
\label{eq:dp}
\begin{align}
\left\{ x_i^\star,u_i^\star \right\}_{i=0}^\infty
=\  \arg \min_{x_i,u_i}\ &\sum_{i=0}^{\infty} l(x_i,u_i;L) 
\\
 {\rm s.t.}\ &
\label{eq:opt_sysdyn}
x_{i+1} = f(x_i,u_i) & \forall\ i\geq 0
\\
\label{eq:opt_const}
& C(x_i,u_i) \leq 0 & \forall\ i\geq 0
\\
\label{eq:opt_init}
& x_0=x(k)
\end{align}
\end{subequations}
with stage cost $l(x_i,u_i;L)$ defined as a parametric function with parameters $L$, constraint set $C(x_i,u_i) \leq 0$, and initial state $x(k)$.
The notation $\left\{ \cdot \right\}_{i=0}^\infty$ is used to indicate indices from $i=0$ to $\infty$.
The goal in this work is to train a predictive model by learning both $l(x_i,u_i;L)$ and $C(x_i,u_i)$ from state and input measurements, which is referred to as the inverse problem to \eqref{eq:dp} in the following.

\subsubsection*{Problem Definition}
The first difficulty in the inverse problem of \eqref{eq:dp} is that measurements $x(k),u(k)$ are not available for $k\rightarrow \infty$ but only in some finite segment. 
We address this using a shortest path formulation (see Section~\ref{sec:DPtrajseg}).
For cases, where the constraint set $C(\cdot,\cdot)$ is unknown, we propose the construction of a candidate constraint set.
The main step of the proposed approach is the derivation of optimality conditions of the shortest path formulation using the candidate constraint set (see Section~\ref{sec:KKTtrajseg}).
The optimality conditions are then used to simultaneously identify constraints from the candidate set and learn the stage cost parameters. 

\subsection{Formulation of infinite-horizon as shortest path problem}
\label{sec:DPtrajseg}
We formulate the infinite-horizon problem as a shortest path problem of finite length $e$ and show that the minimizers of both the infinite-horizon problem and the shortest path problem are identical along the path, i.e., from time $k$ to $k+e$.
Let ${X}^m := 
[\ x(k)^\tp\
x(k+1)^\tp\
\hdots\
x({k+e})^\tp\ ]^\tp \in \mathbb{R}^{n(e+1)}$ and 
${U}^m := 
[\ u(k)^\tp\
u(k+1)^\tp\
\hdots\
u({k+e-1})^\tp\ ]^\tp \in \mathbb{R}^{me}$ be the collection of state and input measurements, respectively, over the time interval $k$ through $k+e$.
If $X^m$, $U^m$ describe the shortest path, then they (at least locally) minimize
\begin{equation}
\label{eq:segment}
\begin{aligned} 
\left\{X^m,U^m\right\} =  \arg  &
 \min_{
x_i,
u_i}\ 
\sum_{i=0}^{e-1} l(x_i,u_i;L) 
\\
&
{\rm s.t.}
\begin{array}[t]{lll}
 x_{i+1} = f(x_i,u_i)  
 \\
C(x_i,u_i) \leq 0\quad  i=0,...,e-1
 \\
x_0 = x(k)
\\
x_e = x(k+e).
\end{array}
\end{aligned}
\end{equation}
Using Bellman's principle of optimality \cite{Bellman1957}, we can show that $X^m$, $U^m$ then also correspond to minimizers of \eqref{eq:dp} for $i=k,...\ k+e$, which is formally stated in the following theorem.

\begin{theorem}
Consider a trajectory segment of measurements ${X^m}$, ${U^m}$ from a dynamical system \eqref{eq:sys}. 
If the observed inputs $U^m$ are the result of the optimal control problem in \eqref{eq:dp} for times $k,...,k+e-1$, then $X^m$, $U^m$ also (at least locally) minimize the optimization problem in \eqref{eq:segment}.
\end{theorem}
\begin{proof}
The optimization problem in \eqref{eq:dp} can be written as
\begin{equation}
\begin{aligned}
\label{eq:dp2}
J^\star(x(k))=\ \min_{x_i,u_i}\ 
&
\sum_{i=0}^{e-1} l(x_i,u_i;L)  + \sum_{i=e}^{\infty} l(x_i,u_i;L) 
\\
 {\rm s.t.}\ &
\eqref{eq:opt_sysdyn},\eqref{eq:opt_const},
\eqref{eq:opt_init}.
\end{aligned}
\end{equation}
If $x_e^\star$ is known, then, using Bellman's principle of optimality \cite{Bellman1957} with $x_e=x_e^\star$, \eqref{eq:dp2} can be formulated as
\begin{equation}
\begin{aligned}
\label{eq:dp4}
J^\star(x(k)) =\ 
\min_{x_i,u_i}\ 
& \sum_{i=0}^{e-1} l(x_i,u_i;L)  + J^\star(x_e^\star) 
\\
\quad {\rm s.t.}\ 
&
\begin{array}[t]{lll}
x_{i+1} = f(x_i,u_i) & i=0,...,e-1
\\
C(x_i,u_i) \leq 0 & i=0,...,e-1
\\
x_0=x(k)
\\
x_e=x_e^\star.
\end{array}
\end{aligned}
\end{equation}
Hence, the minimizers of \eqref{eq:dp} and \eqref{eq:dp4} are equal for all $i=0,...,e$.
The result follows with $x_e^\star=x(k+e)$.
\end{proof}

Note that problem \eqref{eq:segment} differs from a standard finite-horizon formulation as used in \cite{Englert2017} by the end-point constraint $x_e=x(k+e)$, which makes a key difference for learning the problem parameters, as will be illustrated in Section~\ref{sec:simulation}.

\begin{remark}
The shortest path formulation originates from the hypothesis that demonstrations are optimal with respect to the infinite-horizon problem in \eqref{eq:dp}. 
For a different model/ hypothesis, the formulation of the inverse problem can differ.
A particular advantage of the shortest path formulation is that any path along the measured trajectory can be used for learning.
This allows for selecting particular paths where the assumption of optimal execution/data is fulfilled "more closely," e.g., high signal-to-noise ratio or negligible reactive human motor control component in the application considered. 
\end{remark}

\subsection{Optimality conditions} 
\label{sec:KKTtrajseg}
In the following, we derive optimality conditions of the shortest path problem in \eqref{eq:segment} and show how they can be used for {learning} both parameters of the stage cost and constraints.
First, we express the optimization problem in \eqref{eq:segment} in terms of the inputs $u_i$ by recursively defining $x_i=F_i(U,x_0)$:
\begin{align}
\label{eq:FU}
F_i(U,x_0) :=
\begin{cases}
x_0 \quad &{\rm if}\quad i = 0
\\
 f(F_{i-1}(U,x_0),u_{i-1})
\ &{\rm else}
\end{cases}
\end{align}
with
${U} := 
\begin{bmatrix}
u_0^\tp
&
u_1^\tp
&
\hdots
&
u_{e-1}^\tp
\end{bmatrix}^\tp$.
Hence, the resulting optimization problem is given as
\begin{equation}
\label{eq:seg_final}
\begin{aligned}
 \min_{U}\ &
 \sum_{i=0}^{e-1} l(F_i(U,x(k)),u_i;L) 
\\
{\rm s.t.}\
&
\begin{array}[t]{lll}
C(F_i(U,x(k)),u_i) \leq 0\quad   i=0,..., e-1
\\
F_e(U,x(k)) = x(k+e),
\end{array}
\end{aligned}
\end{equation}
where we use $x_0 = x(k)$.
The Lagrangian $\mathcal{L}(U,\lambda,\nu,L)$ of the optimization problem in \eqref{eq:seg_final} is given by
\begin{align}
\begin{aligned}
\label{eq:lagrangian}
&\mathcal{L}( U,\lambda,\nu,L) =  \nu^\tp (F_e(U,x(k)) - x(k+e))
\\
& 
+ \sum_{i=0}^{e-1} l(F_i(U,x(k)),u_i;L) 
+ \lambda_i^\tp {C}(F_i(U,x(k)),u_i)
\end{aligned}
\end{align}
with the Lagrange multipliers $\lambda_i \geq 0$ and $\nu \in \mathbb{R}^n$ (see \cite{Boyd2004}), and $L$ denoting the parameters of the stage cost $l(x_i,u_i;L)$.
Using $\mathcal{L}(\cdot)$ in \eqref{eq:lagrangian}, the KKT optimality conditions for the trajectory segment are given by
\begin{subequations}
\label{eq:kkt}
\begin{align}
&\nabla_U 
\mathcal{L}(U,\lambda,\nu,L)
= 0
\\
&\lambda_i^{ \tp} {C}(F_i(U,x(k)),u_i) = 0 & \hspace{-.6cm} i=0,...,e-1
\\
&\lambda_i\geq 0 & \hspace{-.6cm}  i=0,...,e-1
\\
& 
\label{eq:primal}
C(F_i(U,x(k)),u_i)  \leq  0 & \hspace{-.6cm}  i=0,...,e-1
\\
& 
\label{eq:Fe}
F_e(U,x(k)) - x(k+e)=0.
\end{align}
\end{subequations}

\subsubsection{Construction of candidate constraint set}
\label{sec:cons}
Eq. \eqref{eq:primal} will hold for any observed trajectory with optimal execution (primal feasibility); however, the function $C$ might be unknown.
If $C$ is unknown, we propose to use \eqref{eq:primal} to construct candidate constraints $\bar C(x_i,u_i)$ as the convex hull of all observed data points of the form $P[x_i^\tp\ u_i^\tp]^\tp\leq p$.
A subset of the candidate constraints is then identified as constraints via the KKT conditions. 
A method for computing the convex hull, i.e., $P$ and $p$, is, e.g., presented in \cite{Barber1996}.

\subsubsection{Optimality conditions for learning}
The idea of the proposed approach is to solve \eqref{eq:kkt} for the parameters $L$ of the stage cost $l(x_i,u_i;L)$ as well as for $\lambda_i$ and $\nu$, given measurements ${X^m},\ {U^m}$ and the candidate constraints $\bar C(x_i,u_i)$, i.e.
\begin{subequations}
\label{eq:kktinfer}
\begin{align}
\label{eq:stationarity}
& \nabla_U   
\left. 
\bar{\mathcal{L}}(U,\lambda,\nu,L) 
\right|_{U= {U^m}}=0
\\
\label{eq:complem_ex}
& \lambda_i^{ \tp} {\bar C}(x(k+i),u(k+i)) = 0  &   i=0,...,e-1
\\
\label{eq:complem_dual}
& \lambda_i\geq 0 &   i=0,...,e-1
\end{align}
\end{subequations}
with the approximate Lagrangian $\bar{\mathcal{L}}(\cdot)$ defined as in \eqref{eq:lagrangian} where $\bar C(F_i(U,x(k)),u_i)$ replaces $C(F_i(U,x(k)),u_i)$.
Eq. \eqref{eq:primal} is only needed for the construction of candidate constraints and \eqref{eq:Fe} holds by construction.
Hence, both $\bar C(x(i),u(i))\leq 0$ and \eqref{eq:Fe} are not needed for learning the stage cost parameters [see \eqref{eq:kkt} with \eqref{eq:kktinfer}]. 
The feasibility problem in \eqref{eq:kktinfer} is convex if $l(x_i,u_i;L)$ is linear in $L$.
One can show that \eqref{eq:kktinfer} is always feasible using the convex hull as the candidate constraint set, provided optimal and noise-free data.

The Lagrange multipliers $\lambda_i$ and their values are essential in the proposed IOC approach in order to identify constraints from the candidate set.
Each scalar $\lambda_{i,j}$ can be interpreted as a force keeping the optimization problem \eqref{eq:seg_final} from violating the corresponding primal constraint $\bar C_j(x_i,u_i)\leq 0$ at time $i$.
In other words, the value of a dual variable $\lambda_{i,j}$ indicates the sensitivity of the optimization problem to the corresponding constraint \cite{Boyd2004}.
We define a measure for the identification of constraint $j$ as $\Lambda_j\geq \bar \Lambda$ with
\begin{align}
\label{eq:lambdasum}
\textstyle
\Lambda_j=\sum_{i=0}^{e-1}\lambda_{i,j},
\end{align}
where $\bar \Lambda \geq 0$ is a problem-specific threshold value.
If, e.g., $\Lambda_j=0$, the $j^{\rm th}$ constraint does not affect the minimizer of the optimization problem and does not represent a constraint.
If, however, the value of $\Lambda_j$ is very high, the minimizer is strongly affected by the constraint $j$ and the constraint is therefore crucial in explaining the observed trajectory.
Hence, $\Lambda_j$ relates directly to the importance of constraint $j$. 
The larger $\Lambda_j$, the more important is constraint $j$.
We utilize this relation to identify constraints from the candidate set.
The identified constraints are used in the predictive model, along with the learned parameters of the stage cost.

\subsection{Sub-optimal and noisy data}
\label{sec:prac}
Eq. \eqref{eq:kktinfer} will be feasible if, and only if, the trajectory is the solution of an optimal control problem of the form \eqref{eq:dp}.
In practice, however, even if this modeling assumption is correct, the feasibility problem in \eqref{eq:kktinfer} will not be satisfied exactly due to measurement or process noise.
In order to learn from sub-optimal or noisy data, we propose to solve  the relaxed problem 
\begin{equation}
\label{eq:opt_sub}
\begin{aligned}
 \min_{
L, \nu, \lambda_i}\
&
\left\| 
\left. \nabla_U \bar{\mathcal L}(U,\lambda,\nu,L) \right|_{U=U^m}
\right\|_2^2
\\
\quad {\rm s.t.}\
&
\begin{array}[t]{lll}
\lambda_i^\tp \bar C( x(k+i),u(k+i)) =0 
\\
\lambda_i \geq 0 \quad\quad\quad\quad\quad i=0,...,e-1.
\end{array}
\end{aligned}
\end{equation}
It is easy to verify that 
$\left\| \left. \nabla_U \bar{\mathcal{L}}(\cdot) \right|_{U=U^m}\right\|_2^2=0$ indicates optimality with respect to \eqref{eq:kktinfer} and that \eqref{eq:opt_sub} is always feasible.

\begin{remark} 
The use of a shortest path formulation in this work is reflected through the term $\nu^\tp (F_e(U,x(k)) - x(k+e))$ in \eqref{eq:lagrangian}.
Thus, an inverse approach with finite horizon as in \cite{Englert2017} is obtained with $\nu=0$.
\end{remark}
\begin{remark}[On active and identified constraints]
A constraint $j$ is active if $\bar C_j(x_i,u_i) = 0$ at time $i$.
Using the proposed method for constructing candidate constraints, there are always active candidate constraints.
However, it is important to note that not all active candidates yield $\Lambda_j>0$; it is also possible that candidate $j$ is active, i.e., $\bar C_j(x_i,u_i) = 0$, and $\Lambda_j=0$.
Inversely, $\Lambda_j=0$ does not mean that the candidate $j$ is never active but that the observed trajectory would have been the same with and without candidate $j$.
Hence, candidate constraint $j$ is not identified as constraint if $\Lambda_j=0$.
Section~\ref{sec:simulation} illustrates this concept in a simulation example.
\end{remark}

\section{Illustrative Example}
\label{sec:simulation}
In this section, we illustrate the IOC procedure and highlight its key benefits in simulation for a pendulum with the discrete-time state-space representation:
\begin{align*}
\begin{bmatrix}
x_1(k+1)
\\
x_2(k+1)
\end{bmatrix}
 = 
\begin{bmatrix}
x_1(k) + T_s x_2(k)
\\
x_2(k) - T_s\frac{g}{l}\sin x_1(k)
\end{bmatrix}
+T_s
\begin{bmatrix}
0 \\
\frac{1}{ml^2}
\end{bmatrix}
u(k)
\end{align*}
with $x_1(k)=\theta(t)$ at $t=kT_s$ and $T_s=0.01\rm s$, $g=9.81 \rm m/s^2$, $l=1\rm m$, and $m=1\rm kg$.
$\theta(t)$ is the angle and $u(t)$ is the applied torque in $\rm Nm$, where $|u(t)|\leq \bar u$ with $\bar u= 5\rm Nm$ is assumed to be the available torque.
In the following, we consider an optimal controller of the form \eqref{eq:dp} with constraints $u_{i}\leq 5$ and $-u_{i}\leq 5$ and stage cost $l(x_i,u_i;Q^{\rm gt},r^{\rm gt})=x_i^\tp Q^{\rm gt} x_i+r^{\rm gt} |u_i| +u_i^2$. 
The goal in this example is to learn the constraints and the parameters $Q^{\rm gt}$ and $r^{\rm gt}$.

\subsection{Learning with shortest path and finite horizon methods}
\label{sec:infvsf}
First, we highlight the main differences between the proposed shortest path formulation and two finite-horizon methods, i.e., a method using the KKT conditions similarly as in \cite{Englert2017} and a probabilistic IOC method which uses a likelihood maximization similarly as in \cite{Levine2012}.
The finite-horizon KKT method differs from the presented approach by virtue of the term $\nu^\tp (F_e(U,x(k)) - x(k+e))$ in \eqref{eq:lagrangian} and thus, follows readily with $\nu=0$ (removing the term).
The proposed IOC approach, similarly as the approach in \cite{Englert2017}, yield a convex semi-definite program, which can, e.g., be solved with MOSEK \cite{MOSEK}, whereas the likelihood maximization method yields a non-convex optimization problem, which in this example is solved with a projected gradient descent method.

Figure~\ref{fig:e} shows results with trajectory segments from $t=0$s through $t_e$ generated with $Q^{\rm gt}=I$ and $r^{\rm gt}=0$, where we enforce $Q\succeq 0$.
The middle plot shows that the proposed method only needs a segment from $t=0$s through $t_e\approx 0.5$s to find the ground truth. 
Both methods with finite horizon are not able to learn the ground truth even if the segments are long and $\theta(t)$ is close to stationarity (see $Q_{12} \approx 1$ at $t_e=1000$s).

\subsection{Learning with and without candidate constraints}
\label{sec:convsunc}
Next, consider the trajectories with $Q^{\rm gt}=10 I$ and $r^{\rm gt}=1$ for comparing methods with and without candidate constraints using segments from $t_i$ to $t_i+2$s [see Figure~\ref{fig:e200} (top)].

\subsubsection*{IOC, constrained (2nd plot from the top)}
The first step is to construct candidate constraints for the input $u(k)$:
\begin{subequations}
\label{eq:g}
\begin{align}
\label{eq:gu}
u(k)\leq g_u
\\
\label{eq:gl}
-u(k)\leq g_l
\end{align}
\end{subequations}
where $g_u$ and $g_l$ depend on the chosen segment and are displayed in red (diamond markers) and green (triangle markers), respectively.
The algorithm returns $Q$ and $r$ as well as $\Lambda_1$ and $\Lambda_2$, which are defined in \eqref{eq:lambdasum} and correspond to the candidate constraints \eqref{eq:gu} and \eqref{eq:gl}, respectively.
The parameters $Q$ and $r$ are very close to the ground truth for all $t_i$.
If $t_i<0.96$s, $g_u=5$ and $\Lambda_1>0$ suggesting that $u(k)\leq 5$ is indeed a constraint.
If $t_i>0.96$s, $g_u<5$ and $\Lambda_1=0$ suggesting that $u(k)\leq g_u<5$ is not a constraint, which is correct, as the constraint is not active.
For all $t_i$, $g_l<5$ and $\Lambda_2=0$ (not displayed) suggesting that $-u(k)\leq g_l<5$ is not a constraint.
Overall, $Q$ is learned reliably and for $t_i<0.96$, $u(k)\leq u_{max}$ is learned as constraint.
The trajectory does not provide conclusive evidence about the existence of a lower bound, i.e., $-u(k)\leq u_{max}$, which is expected as $g_l<5\ \forall t_i$.

\subsubsection*{IOC, unconstrained (3rd plot from the top)}
If $t_i>0.96$s, $Q$ and $r$ are very close to the ground truth, which is expected since the control problem is virtually unconstrained in these segments.
However, if no candidate constraints are constructed a priori, $Q$ and $r$ differ for $t_i<0.96$s as the observed trajectory cannot be explained by means of an unconstrained optimal control problem.

\subsubsection*{Finite-horizon IOC, constrained (bottom plot)}
The method learns the constraint $u(k)\leq 5$ using similar arguments as the proposed shortest path IOC method; however,
it fails to capture the ground truth stage cost parameters with $r\approx 0$ and $Q$ not close to $Q^{\rm gt}$ for all trajectory segments. 
\begin{figure}[t]
      \centering
     \includegraphics[width=0.96\columnwidth]{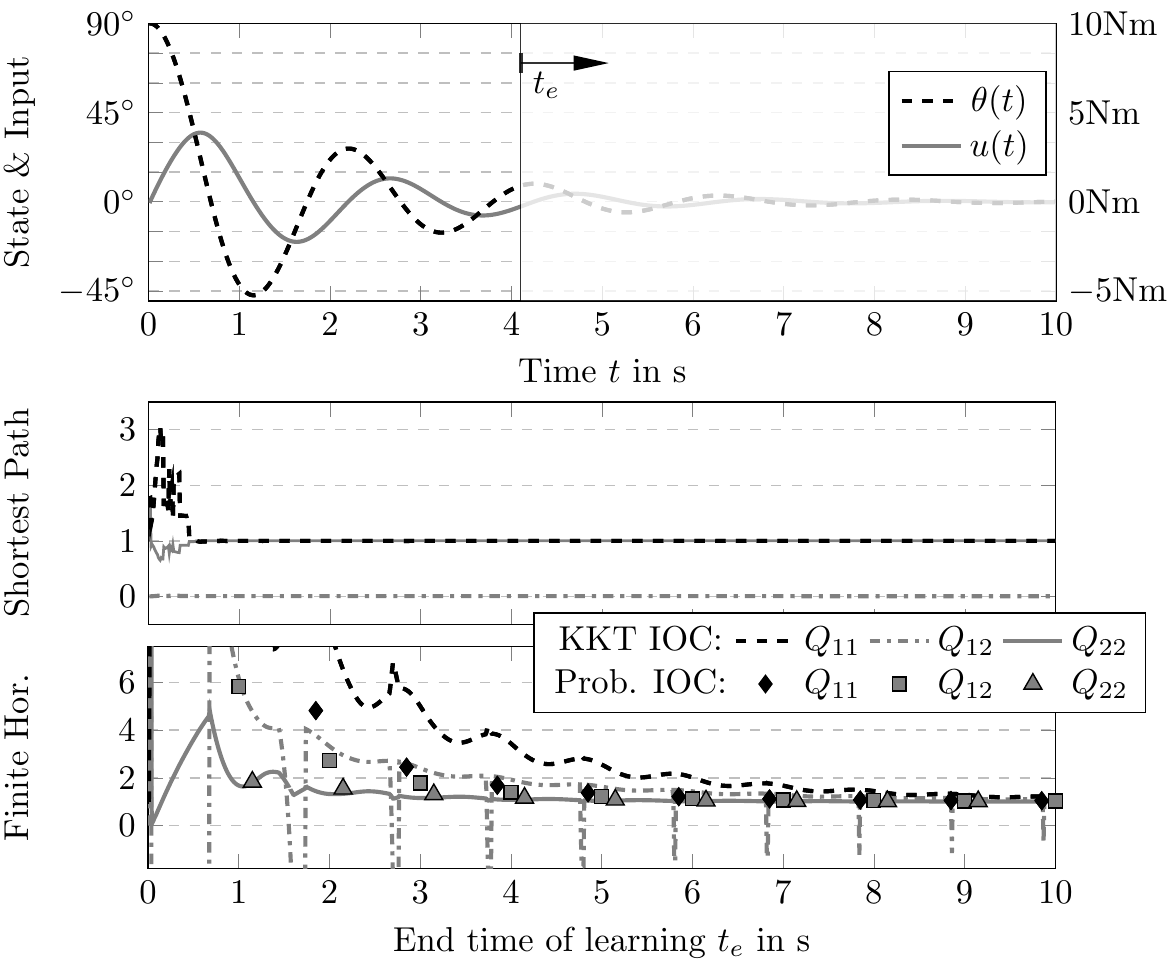}
      \caption{{Top: State and input trajectories. 
      Middle: $Q$ learned with shortest path IOC. 
      Bottom: $Q$ learned with two finite-horizon methods: KKT and maximum likelihood.}}
      \label{fig:e}
\end{figure}
\begin{figure}[t]
      \centering
     \includegraphics[width=0.96\columnwidth]{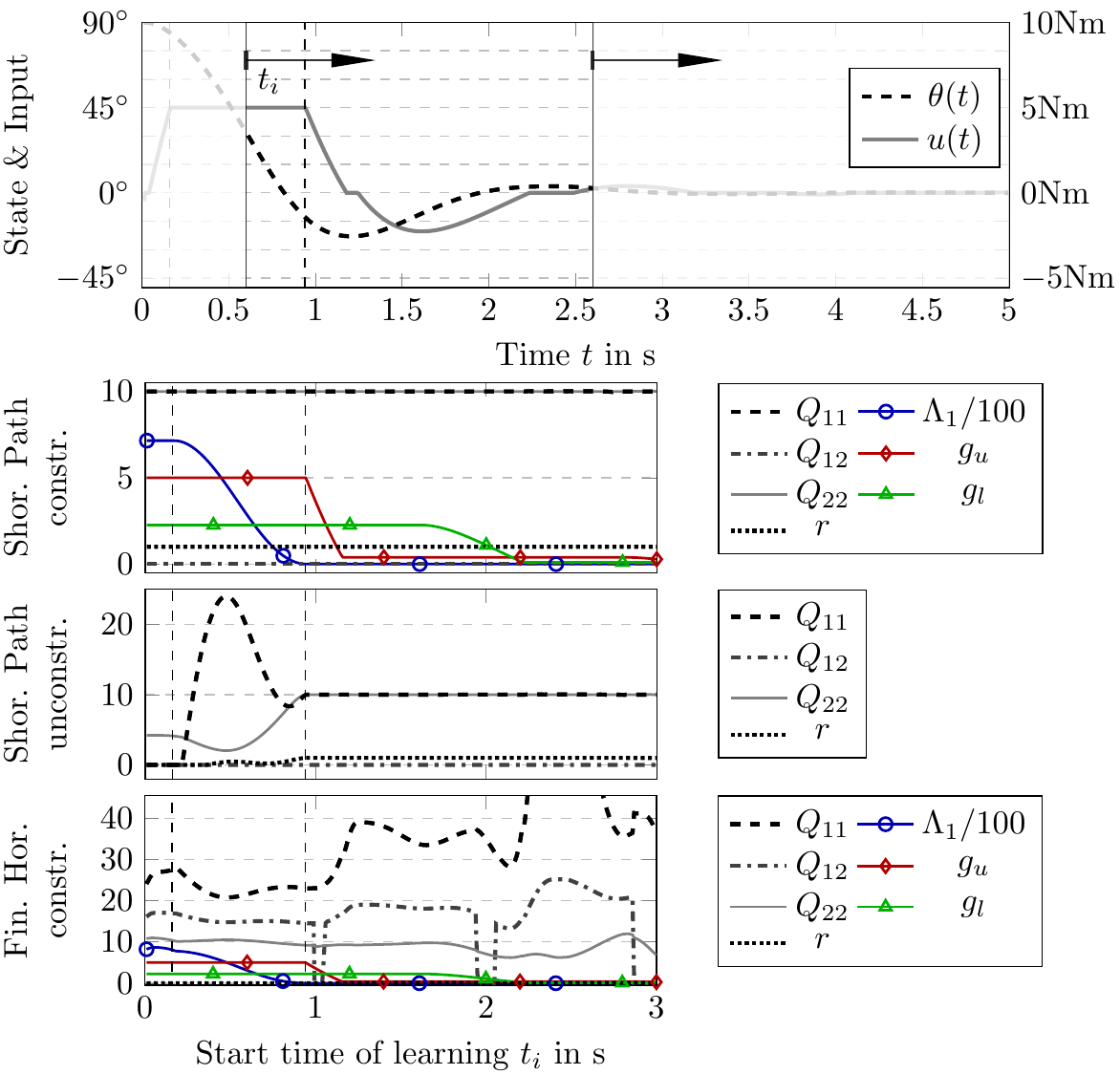}
      \caption{
Top: State and input trajectories. 
2nd from the top: Parameters learned with the proposed method with candidate constraints. 
3rd from the top: $Q$ learned without candidate constraints. 
Bottom: Parameters learned with finite-horizon method with candidate constraints. 
      }
      \label{fig:e200}
\end{figure}
\subsection{Summary of analysis}
In this section, we have illustrated the benefits of the proposed approach.
In particular, we showed the candidate constraint construction and how to simultaneously learn parameters of the stage cost and identify constraints from the candidate set.
Further, we have shown that the proposed shortest path formulation only requires a short segment of measurements to learn the stage cost parameters and identify constraints, whereas finite-horizon approaches require a comparably long segment. 
Moreover, we have shown the importance of the candidate constraint set as a substantial component for correctly identifying the stage cost. 

\section{Manipulation of a Passive Kinematic Object}
\label{sec:robot}
In this section, we show how to train a predictive model for human movements in a manipulation task using the proposed method. 
We conducted experiments with three human subjects where the underlying hypothesis is that humans plan their movements by solving a constrained optimal control problem.

\subsection{Experiment description and system modeling} 
In the experiment, the human subjects manipulated one end of an object whose position was changed consecutively by a robot.
The manipulation task was set up to provide a foreseeable environment triggering the human's predictive motor control component such that the reactive control component can be disregarded (at least at the beginning of the movement).
The object was articulated and unactuated and was composed of three lightweight wooden links and one cardboard handle, which acted as both a revolute joint and the manipulation point (see Figure~\ref{fig:hri}). 
Hence, it had four revolute joints, one connecting its end link to the robot (joint 1), two connecting the three wooden links (joint 2 \& 3), and 
the cardboard handle (joint 4), which was gripped by the subject such that the forearm and the handle acted as a single rigid body. 

After familiarizing themself with the robot, the human was instructed to achieve specific angles for two of the object's joints, 
the joint connecting the object to the robot  (joint 1 in Figure~\ref{fig:hri}) and the first joint after that (joint 2), 
both of which have vertical rotational axes (perpendicular to the ground).
The target angles were communicated to the subjects visually by reference-markers attached to the links.
The subjects were asked to only move when the robot was stationary. 
First, the robot moved to disturb the system state.
When the robot's motion ended, the subject corrected the reference error.
Motion capture sensors were placed on all links of each kinematic chain and recorded through the Phasespace Python API.

The derivation of the individual movement model, i.e., the system dynamics, of each subject is based on modeling the passive kinematic object and the human arm as a kinematic chain \cite{isb1} whose parameters were identified from measurements. 
In this model, the base frame is attached to the torso and the manipulation frame is attached to the grip location of the hand. 
Ball joints such as the shoulder joint are modeled as three revolute joints in series with orthogonal axes intersecting at the center of the joint. 
This leads to the ball joint configuration being described with intrinsic Euler angles rotating around a point in space \cite{shoulder,ball}. 
The elbow joint is modeled as a single revolute joint. 
The wrist is modeled as three revolute joints in series; however a wrist brace was used in the experiment to restrict the motions in the frontal and sagittal plane, that is, waving and flapping motions. 
Pronation and supination (twisting about the forearm) could not be restricted by the brace; however the experiment was designed such that the kinematic chain of the object itself constrained this movement.
Both the placement of the motion capture markers and the kinematic modeling are shown in Figure~\ref{fig:mocap}.
\begin{figure}[t]
      \centering
     \includegraphics[width=1\columnwidth]{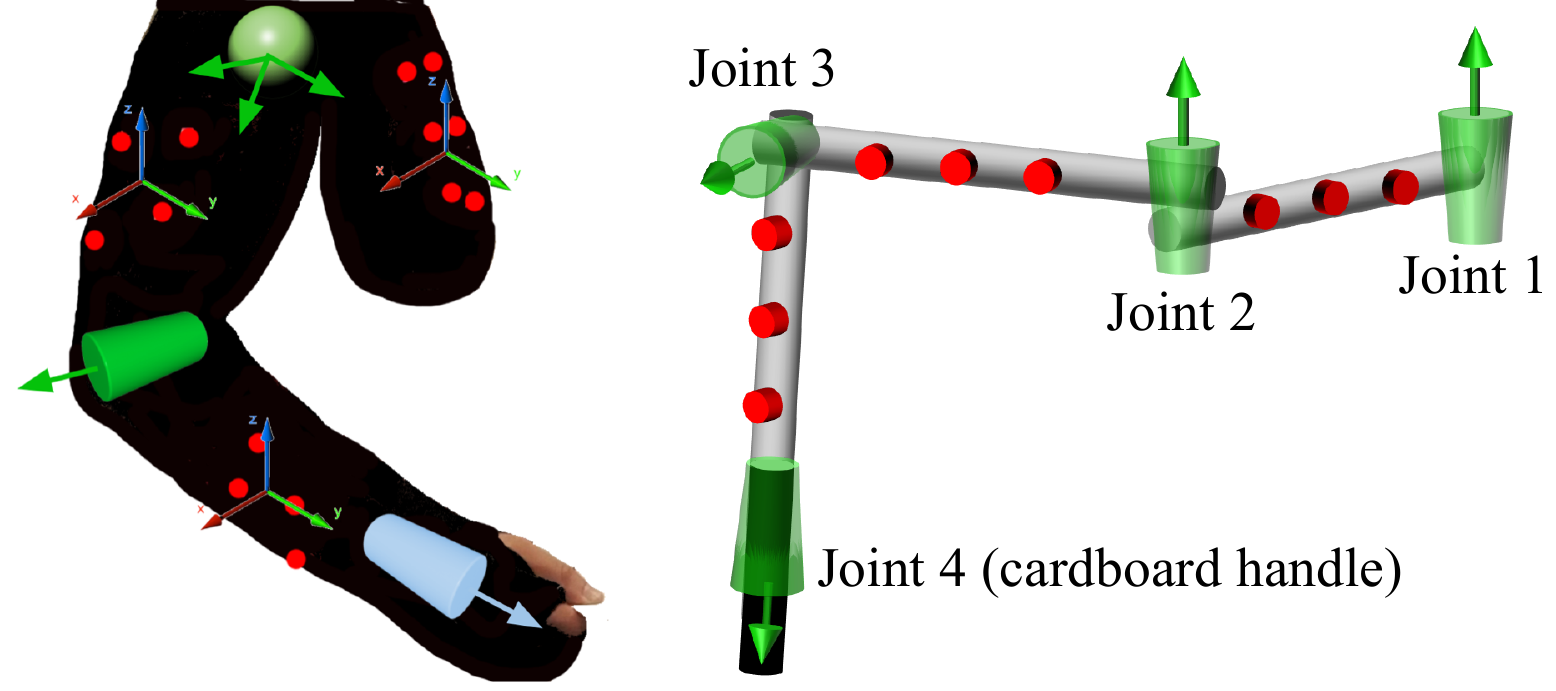}
  \\      \includegraphics[width=0.825\columnwidth]{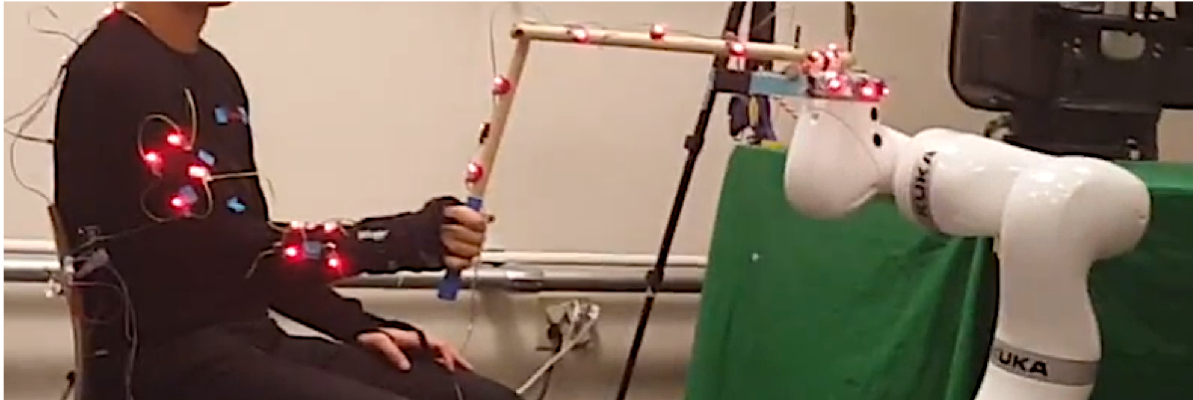}
      \caption{
      Top: Modeling of the human arm and the object. 
      Bottom: Experiment setup with the Kuka LBR iiwa robot.
      Joints included in the model are shown in green, while the blue joint represents a freedom of motion that was constrained by experiment design. 
      The motion capture markers are illustrated in red.
      }
      \label{fig:hri}
      \label{fig:mocap}
\end{figure}

The system state $
x(t) = 
[\
x_h (t)^\tp\
x_o (t)^\tp\
]^\tp$ is composed of the joint angles of the human, $x_h(t) \in \mathbb{R}^4$, and of the object, $x_o(t) \in \mathbb{R}^4$.
The input to the system, 
$u(t)=\dot{x}_h(t)$,
is given by the joint velocities of the human arm. 
The velocities of the object joint angles are given by:
\begin{equation}
\dot{x}_o(t) = J^{\ddagger}_{o}(x_o(t)) V_{g}(t), 
\label{eq:xrdot}
\end{equation}
where $J_{o}(x_o(t)) \in \mathbb{R}^{6 \times 4}$ is the Jacobian mapping joint velocities of the object to $V_g(t)$, the absolute twist velocity of the manipulation frame, and $J^{\ddagger}_{o}(x_o(t)) \in \mathbb{R}^{4 \times 6} $ denotes its Moore-Penrose pseudo-inverse \cite{inverse}. 
Given that the human maintained a stationary base in the experiment, we can express $V_g(t)$ in terms of the human arm joint velocities and the Jacobian of the human arm, $J_h(x_h(t))  \in \mathbb{R}^{6 \times 4}$:
\begin{equation}
V_g(t) = J_{h}(x_h(t))\dot x_h(t). 
\label{eq:griptwist}
\end{equation}
Using \eqref{eq:xrdot} and \eqref{eq:griptwist}, $\dot{x}_o(t) = J^{\ddagger}_{o}(x_o(t)) J_{h}(x_h(t))\dot x_h(t)$, and thus, the overall dynamics of the system is given by
\begin{equation}
\begin{bmatrix}
\dot x_h (t)
\\
\dot x_o (t)
\end{bmatrix}
=
\begin{bmatrix}
I
\\
J^\ddagger_o(x_o(t)) J_h(x_h(t))
\end{bmatrix}
u(t). \label{eq:exp_sys_desc}
\end{equation}
In order to obtain the Jacobians, the twists representing the joints in each kinematic chain are identified by recording traces of the subject's range of motion and applying the techniques in \cite{iros15}.
The Jacobians $J_h(x_h(t))$ and $J_o(x_o(t))$ in \eqref{eq:exp_sys_desc} are derived using the formula for the body Jacobian as in \cite{math}. 

A discrete-time representation of \eqref{eq:exp_sys_desc} is derived using an Euler-forward scheme with the sampling time $T_s$:
\begin{align*}
\begin{bmatrix}
x_h (k+1)
\\
 x_o (k+1)
\end{bmatrix}
=
\begin{bmatrix}
x_h (k)
\\
 x_o (k)
\end{bmatrix}
+
T_s
\begin{bmatrix}
I
\\
J^\ddagger_o(x_o(k)) J_h(x_h(k))
\end{bmatrix}
u(k).
\end{align*}
An unscented Kalman filter as described in \cite{ukf} is implemented to estimate the system state, where a static process model is chosen to smoothen the estimated angles, since measurement noise is amplified by the kinematic transformation.
The inputs are computed as $u(k)=(x_h(k+1)-x_h(k))/T_s$.

\subsection{Learning predictive model for human movements}
\label{sec:ex}

Each of the three subjects maneuvered the object 15 times to correct the reference error induced by the robot.
For each experiment, we recorded the entire trajectory from the start of the human movement until the subject was instructed to remain stationary.
For reasons discussed in Section~\ref{sec:compareFI}, we use the initial $1.2$s, i.e., $e=65$ in \eqref{eq:kktinfer} with sampling time $T_s=0.0185$s for learning, which corresponds to roughly 60\% of each trajectory.
In order to generalize from the available sparse data, we utilize leave-one-out cross-validation \cite{James2013}, where we learn the parameters of the predictive model 15 times, each time removing one of the recorded trajectories. 
This is done to assess the robustness of the model. 
 
\subsubsection{Design choices}
\label{sec:12norm}
In this work, we train a predictive model with quadratic stage cost.  
Our goal is to exemplify the proposed method to build a simple predictive model of human movement.
Quadratic stage costs are commonly used as objective function in optimal control offering a good compromise between complexity and expressivity, where the cost minimizes a trade-off between tracking a given target and control effort.
Note that higher-order or more complex stage cost terms are possible with the proposed framework and there are various possibilities to express human movements \cite{Oguz2018b}.
Given that the task requires tracking a reference for only two of the states, we take a stage cost of the form
\begin{align*}
l(x_i,u_i) =\ & (Sx_i - y_s)^\tp Q (Sx_i - y_s) + u_i^\tp Ru_i,
\end{align*}
where $y_s \in \mathbb{R}^2$ is the reference, $S=[\ 0_{2\times 4}\ I_2\ 0_{2\times 4}\ ]$ selects the states (two joint angles of the object) tracking $y_s$, and $Q,R$ are the penalty parameters.
We enforce $Q,R\succeq 0$ in order to obtain physically meaningful penalties for both deviation to the target angles and control effort. 
Also, we restrict the input penalties to $\sum_{i=1}^m R_{ii}=1$, which fixes the scaling of the stage cost and avoids the trivial solution of all parameters being zero.
We train one predictive model without constraints and one with a polytopic candidate constraint set for each subject.

\subsubsection*{Candidate constraints}
The object's states $x_o(k)$ are modeled as unconstrained. 
The human's states $x_h(k)$ consist of the three shoulder joint angles and the elbow angle; the inputs $u(k)$ are the three angular velocities of the shoulder joint and the angular velocity of the elbow. 
Constraints  on joint angles directly relate to constraints on $x_h(k)$, velocity constraints relate to constraints on $u(k)$, and acceleration constraints are computed as a rate constraint: $a(k) = (u(k+1)-u(k))/T_s$. 

\subsubsection{Learning results}
\label{sec:infres}
Figure~\ref{fig:results} shows the mean and standard deviation of $Q$ and $R$ obtained with the proposed IOC method. 
The most distinct feature is the scale of the parameters $Q_{ij}$, varying from order $10^{-2}$ for Subject~1, $10^{-3}$ for Subject~2, to $10^{-6}$ for Subject~3.
The second most distinct feature is the difference in the diagonal elements of $R$ that reflect movement of the shoulder, i.e., $R_{11}$, $R_{22}$, and $R_{33}$, whereas the penalty on elbow velocity is comparable, i.e., $R_{44}\approx 0.2$ for all subjects.
Off-diagonal elements in $R$ are similar across subjects.

Table~\ref{tb:lagrange} shows the sum of Lagrange multipliers as in \eqref{eq:lambdasum}, which are used to identify constraints from the candidate constraint set.
The Lagrange multipliers are stated as the mean over all experiments to identify constraints on angle, velocity, and acceleration of shoulder and elbow joints. 
We consider constraint $j$ as identified if the corresponding Lagrange multiplier $\Lambda_j\geq \bar \Lambda = 10^{-3}$.
It can be seen that constraints are predominantly on shoulder movement.
Constraints on elbow movement seem less important for all subjects.
Note that even though the stage cost parameters in Figure~\ref{fig:results}	 obtained with constrained and unconstrained IOC are relatively close for the individual subject, the resulting prediction models differ by virtue of the constraints identified as in Table~\ref{tb:lagrange}.
\begin{figure}[htb]
      \centering
     \includegraphics[width=1\columnwidth]{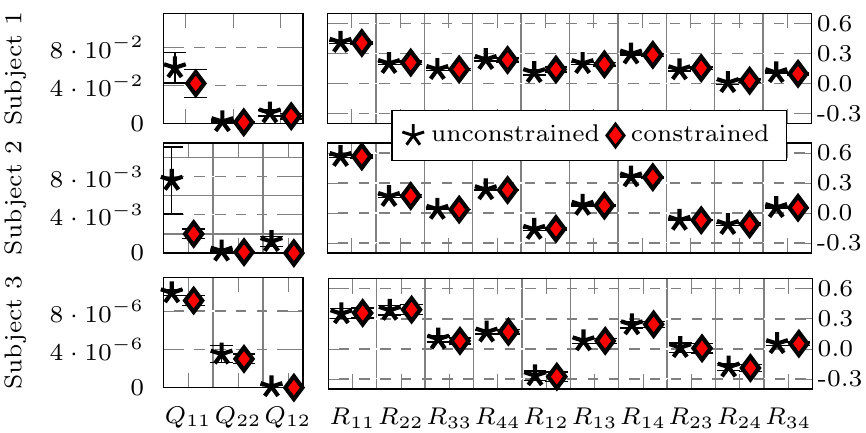}
      \caption{
      Mean and standard deviation of cost parameters $Q$ and $R$ for unconstrained learning (black stars) and constrained learning (red diamonds). 
      }
      \label{fig:results}
\end{figure}
\begin{table}[tbh]
\setlength\tabcolsep{3.9pt} 
\caption{
  Lagrange multipliers to identify constraints
}
\begin{center}
\def\arraystretch{1.05}
\label{tb:lagrange}
\begin{tabular}{llllll llll}
\hline
&  \multicolumn{2}{c}{Angle} & \multicolumn{2}{c}{Velocity} & \multicolumn{2}{c}{Acceleration} \\
& Shoulder  & Elbow & Shoulder & Elbow  & Shoulder & Elbow\\
\hline
Subject 1
& $22.8$ & $0$ & $3.31e$-$2$ & $0$  & $1.38e$-$2$ & $8.66e$-$4$
\\
Subject 2
& $11.5$ & $0$ & $2.78e$-$1$ & $0$ & $2.15e$-$2$ & $6.98e$-$4$
\\
Subject 3 
& $3.50$ & $0$ & $4.36e$-$1$ & $2.86e$-$4$ & $1.05e$-$1$ & $0$
\\
\hline
\end{tabular}
\end{center}
\end{table}

\subsection{Evaluation of trained human manipulation model}
\label{sec:evalinf}
The difficulty in evaluating the quality of the trained model for human-centered experiments is the lack of a ground truth as reference.
We therefore assess the quality of modeling human movement as an optimal control problem \eqref{eq:dp} by comparing the true trajectory with the prediction provided by the model. 
The predictions are obtained by solving problem \eqref{eq:segment} with the learned stage cost and identified constraints from the initial position at time $t=0$s through $t=t_e=1.2$s using IPOPT \cite{Biegler2009} (see Figure~\ref{fig:graph_con} for a sample prediction).
We compute 15 sets of stage cost matrices by leaving out one trajectory for each learning. 
In order to evaluate the quality of the trained model, we use the left-out measured trajectory for validation against the predicted trajectory, which would result from \eqref{eq:segment} with the  learned stage cost and constraints.
This technique ensures that the predicted trajectory is not biased by the corresponding measured trajectory.
The mismatch between prediction $\hat x_i^j\in \mathbb{R}^8$ and measurement $x^j(i)\in \mathbb{R}^8$ of trajectory $j$ is measured as the \underline root \underline mean \underline square (RMS) error:
\begin{align}
\label{eq:error}
\textstyle
E^j
=
\sqrt{
\frac{1
}{
8e}
\sum_{i=1}^{e}
\|
\hat x_i^j - x^j(i)
\|_2^2
}.
\end{align}

\subsubsection{Intra-subject evaluation}
First, we compute the errors $E^j$ in \eqref{eq:error} for each trajectory $j$ per subject. 
Figure~\ref{fig:graph_con} shows one measured trajectory of Subject~2 and the predictions obtained with the unconstrained and the constrained model.
The prediction obtained with the unconstrained model shows a larger RMS error, best seen in the plot of human joint angles.
The prediction obtained with the constrained model shows a lower error. 
Table~\ref{tb:graph_con} presents the mean and standard deviation over all 15 prediction errors for all subjects.
It shows that, generally, the predictions have low errors ($<3.3^\circ$), where Subject~1 has the lowest ($<1^\circ$).
On average, the presence of constraints improve the predictions by 20\%-25\%. 
\begin{table}[h]
\setlength\tabcolsep{8.2pt}
\caption{
Prediction errors: Unconstrained vs. constrained
}
\begin{center}
\def\arraystretch{1.05}
\label{tb:graph_con}
\begin{tabular}{l cccccclllll llll}
\hline
Constraint set & unconstrained  & constrained 
 \\
\hline
Subject 1 
& $0.96^\circ \pm 0.49^\circ$   
& $0.78^\circ \pm 0.42^\circ$ 
\\
Subject 2 
& $3.26^\circ \pm 1.75^\circ$      
& $2.45^\circ \pm 0.87^\circ$  
\\
Subject 3 
& $1.87^\circ \pm 1.00^\circ$ 
& $1.56^\circ \pm 0.79^\circ$  
 \\ 
\hline
\end{tabular}
\end{center}
\end{table}

\subsubsection{Inter-subject cross-evaluation}
Next, we analyze the individuality of the trained models, where the error $E^j$ in \eqref{eq:error} is computed three times for each trajectory $j$:
We compute the error using the prediction model of the subject who generated trajectory $j$; 
then, we compute $E^j$ of the predicted trajectory $\hat x_i^j$ using the other subjects' prediction models, where we use the proposed IOC method with polytopic constraints.

Figure~\ref{fig:inter} shows an example of a measured trajectory from Subject~1, compared against predictions generated with the models of all subjects. 
The measured trajectory and the predicted trajectory of Subject~1 are close (error: $0.55^\circ$).
The predicted trajectories of Subject~2 \& 3 show higher errors.
Table~\ref{tb:inter} states the mean and standard deviation of the errors between measurements of Subject~$j$ in columns $j$ and prediction with objective of Subject~$i$ in rows $i$ over all trajectories.
Hence, a good separation between the subjects means large entries in the off-diagonal entries $i\neq j$.
The results show a high confidence in separating Subject~1 from the other two with high confusion errors ($3.23^\circ$, $2.39^\circ$ vs. $0.78^\circ$). 
The confidence to identify Subject~2 from a given trajectory is also high with confusion errors ($3.99^\circ$, $3.59^\circ$ vs. $2.45^\circ$). 
A less clear separation is observed for Subject~3, where the confusion errors are lower ($2.22^\circ$, $1.91^\circ$ vs. $1.56^\circ$). 
Overall, this cross-validation suggests that the models trained to predict the distinct motor behavior are individual.
\begin{figure}[!t]
      \centering
     \includegraphics[width=0.87\columnwidth]{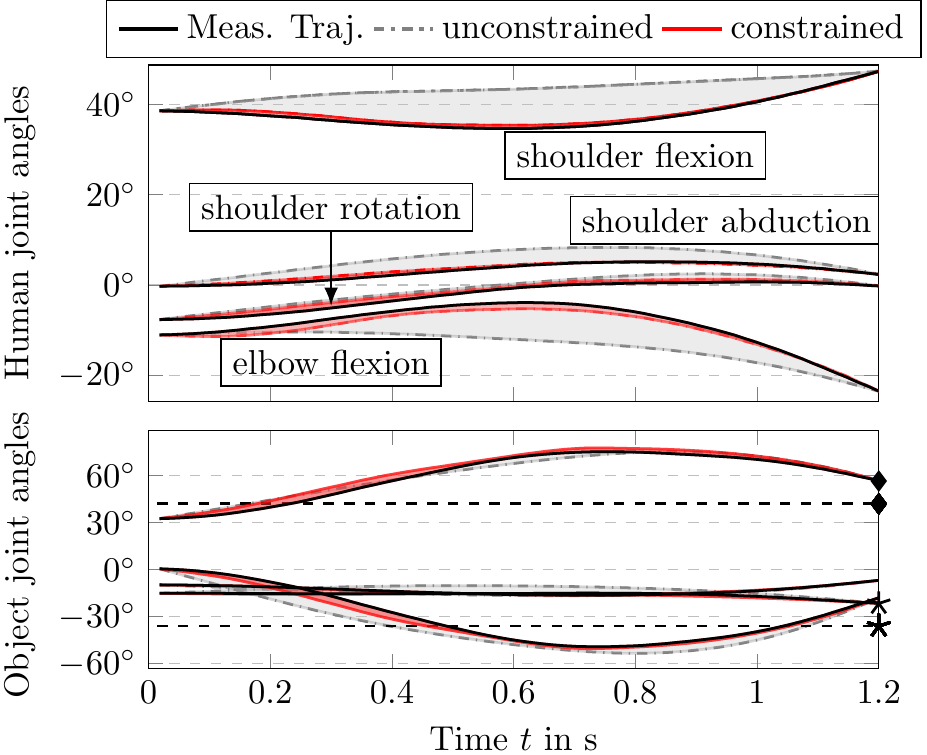}
      \caption{
      Measured trajectory in black, predicted trajectory with the unconstrained model in gray (error $4.17^\circ$) and the constrained model in red (error $1.40^\circ$). 
The upper plot shows the shoulder flexion, shoulder abduction, and shoulder rotation, as well as elbow flexion.
The object states to be tracked are shown in the lower plot as dashed black lines and are related to the corresponding joints with a diamond and a star marker.
      }
      \label{fig:graph_con}
\end{figure}
\begin{figure}[!t]
      \centering
     \includegraphics[width=0.9\columnwidth]{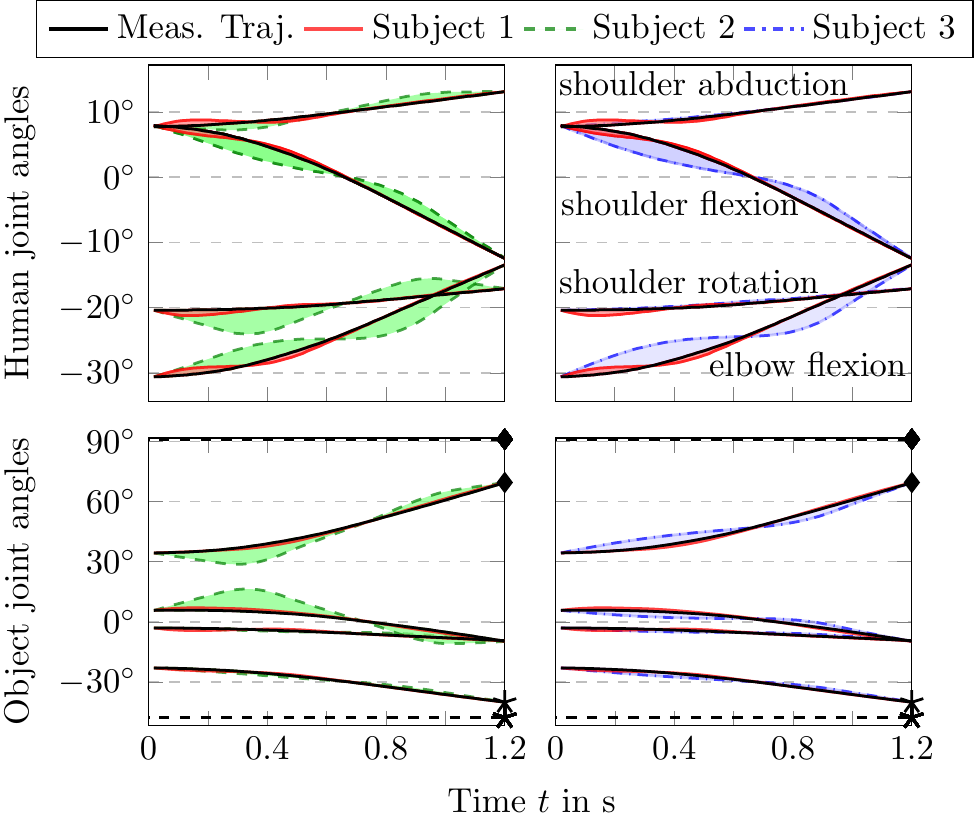}
      \caption{
      Measured trajectory of Subject~1 in black, predicted trajectory of Subject~1 in red (error: $0.55^\circ$).
      Left plots: Predicted trajectory of Subject~2 in green (error: $3.62^\circ$). 
      Right plots: Predicted trajectory of Subject~3 in blue (error: $1.97^\circ$). 
      }
      \label{fig:inter}
\end{figure}
\begin{table}[!h]
\caption{
Prediction errors: Cross-validation between subjects
}
\begin{center}
\def\arraystretch{1.05}
\label{tb:inter}
\begin{tabular}{ccccccc llllll llll}
\hline
\multicolumn{2}{l}{Trajectories of } & Subject 1  & Subject 2 & Subject 3 \\
\hline
\multirow{3}{*}{\rotatebox{90}{Model}} 
& 
Subject 1 &  
$0.78^\circ \pm 0.42^\circ$ & $3.99^\circ\pm 1.53^\circ$ & $2.22^\circ\pm 1.14^\circ$ 
\\
&
Subject 2 &  
$3.23^\circ\pm 1.03^\circ$ & $2.45^\circ\pm 0.87^\circ$ & $1.91^\circ\pm 0.93^\circ$
\\
&
Subject 3 &  
$2.39^\circ\pm 0.68^\circ$  & $3.59^\circ\pm 1.68^\circ$ & $1.56^\circ\pm 0.79^\circ$
\\
\hline
\end{tabular}
\end{center}
\end{table}

\subsubsection{Benefit of shortest path formulation}
\label{sec:compareFI}
In the following, we discuss the advantages of using a shortest path formulation over a finite horizon in the context of the considered application.  
If the entire trajectory is used for training and stationarity is reached, i.e., $e$ is large, both the proposed shortest path method and a finite-horizon method are similar.
In the context of the considered application, however, we encountered two main challenges when considering the entire trajectory. 
Firstly, in the final part of the trajectory, the target angles are more or less reached and the measured signals are close to stationarity.
As a result, the signal-to-noise ratio is low and can corrupt learning. 
Secondly, we observed small corrections around the target angles in the experiment suggesting the presence of reactive movements, which renders the final part of the trajectory not indicative of the predictive human motor control.

For shorter segments, the predictive component dominates both noise and reactive component but the solution from a finite-horizon formulation diverts from that with a shortest path (see Section~\ref{sec:simulation}).
The proposed IOC approach allows for using only the initial part of the trajectory for learning where stationarity is not reached.
Overall, the presence of both reactive human motor control component and noise do not fulfill the assumption of optimal execution with respect to \eqref{eq:dp}.
We used the initial 60\% of the trajectory, which was observed to be a good trade-off between segment-length and avoidance of the reactive component.

Figure~\ref{fig:e_study} revisits the trajectory in Figure~\ref{fig:graph_con} to illustrate the above discussion on the horizon length $e$.
The upper plot shows the complete recorded trajectory, where some correction around the target angles can be observed for $t\geq 1.4$s (see joint angle marked by the diamond symbol).
The lower plot displays the RMS error \eqref{eq:error} of the predictions that result from different horizon lengths $e$.
The RMS error increases as a result of both the correction around the target angles and the low signal-to-noise ratio.
It highlights that the modeling assumption as an open-loop optimal control problem is suitable for the predictive part, but not in the presence of the reactive component.
\begin{figure}[h]
      \centering
     \includegraphics[width=0.82\columnwidth]{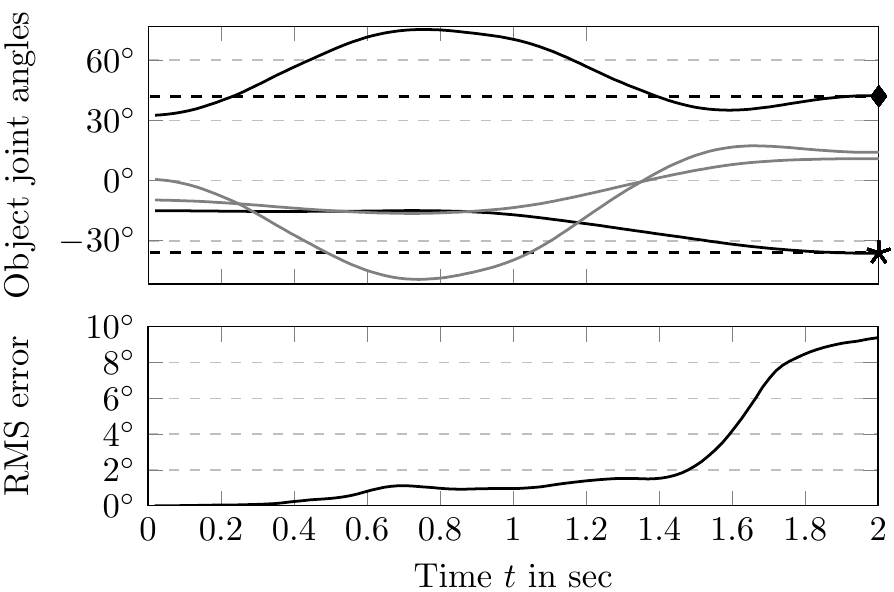}
      \caption{
      Top: Target angles to be tracked are shown as dashed black lines and are related to the corresponding joints with a diamond and a star marker.
Bottom: RMS error of prediction with different horizon lengths $e$.
      }
      \label{fig:e_study}
\end{figure}

\section{Conclusion}
\label{sec:conclusion}
This paper presented an inverse optimal control approach to learn both cost function parameters and constraints from demonstrations, i.e., state and input measurements of dynamical systems.
The shortest path formulation is shown to be the inverse problem to an infinite-horizon optimal control problem. By relying on the Karush-Kuhn-Tucker conditions, the problem is convex for cost functions that are linear in their parameters.
We set up a human manipulation experiment to exemplify the proposed approach for modeling and predicting human arm movements. 
In the experiment, three human subjects manipulated one end of a passive kinematic object whose position was changed consecutively by a robot. 
The benefits of using a shortest path formulation and the consideration of constraints on human movements were highlighted.
The results showed that a model with good predictive capabilities can be learned using a quadratic cost function for both states and inputs together with constraints on shoulder movements using the proposed formulation. 
Finally, it was shown that the predictive models of the human subjects are individual. 


\bibliographystyle{IEEEtran}
\bibliography{root.bbl}

\end{document}